\documentclass[10pt]{article}

\usepackage{amsmath}
\usepackage{amssymb}
\usepackage{amsthm}
\usepackage{latexsym}
\usepackage{color}
\usepackage{graphicx}
\usepackage{color}
\usepackage{mathrsfs}

\usepackage[symbol]{footmisc}

\DeclareSymbolFont{calletters}{OMS}{cmsy}{m}{n}
\DeclareSymbolFontAlphabet{\mathcal}{calletters}

\newtheorem{Theorem}{Theorem}[part]
\newtheorem{Definition}{Definition}[part]

\newtheorem{Lemma}{Lemma}[part]

\newtheorem{Remark}{Remark}[part]
\newtheorem{Example}{Example}[part]

\newcommand{\be}{\begin{equation}}
\newcommand{\ee}{\end{equation}}
\newcommand{\bea}{\begin{eqnarray}}
\newcommand{\eea}{\end{eqnarray}}
\newcommand{\beas}{\begin{eqnarray*}}
\newcommand{\eeas}{\end{eqnarray*}}

\newcommand{\brak}[1]{\left(#1\right)}    
\newcommand{\crl}[1]{\left\{#1\right\}}   
\newcommand{\edg}[1]{\left[#1\right]}     

\newcommand{\nc}{\newcommand}
\nc{\esssup}{\mathop{\mathrm{ess\,sup}}}
\nc{\essinf}{\mathop{\mathrm{ess\,inf}}}
\nc{\argmax}{\mathop{\mathrm{arg\,max}}}

\def \P{\mathbb{P}}

\def \R{\mathbb{R}}

\def \E{\mathbb{E}}

\def \Q{\mathbb{Q}}

\def \1{\mathds{1}}

\def \Ac{{\cal A}}

\def \Cc{{\cal C}}

\def \Ec{{\cal E}}
\def \Fc{{\cal F}}

\def \Pc{{\cal P}}

\def \Nc{{\cal N}}

\def \Qc{{\cal Q}}

\def \l({{\left (}}
\def \r){{\right )}}
\def \l[{{\left [}}
\def \r]({{\right ]}}

\def\Acr{\mathscr{A}}
\def\Ccr{\mathscr{C}}

\def \dint{{\displaystyle \int}}

\DeclareMathOperator*{\argmin}{arg\,min}
\newcommand{\p}{\mathbb{P}}
\newcommand{\q}{\mathbb{Q}}

\newcommand{\MBFigure}[6]{
$\left. \right.$ \\
\refstepcounter{figure}
\addcontentsline{lof}{figure}{\numberline{\thefigure}{\ignorespaces #5}}
\begin{center}
\begin{minipage}{#1cm}
\centerline{\includegraphics[width=#2cm,angle=#3]{#4}}
\begin{center}
\upshape{F\textsc{ig} \normal
\end{center}
size{\thefigure}. $-$} #5
\end{center}
\label{#6}
\end{minipage}
\end{center}
$\left. \right.$ \\}
\makeatletter
\let\@fnsymbol\@arabic
\makeatother
\title{Optimal investment and consumption under $g$- expected utility and general constraints in incomplete market }
\author{Wahid {\sc Faidi}
\thanks{Department of Mathematics, College of Science and Humanities, Shaqra University,
Al Quwayiyah, Saudi Arabia }\;
 \thanks{University of Tunis El Manar, LAMSIN, Tunis, Tunisia} 
\\ e-mail:  \textcolor[rgb]{0.00,0.07,1.00}{faidiwahid@su.edu.sa}- \textcolor[rgb]{0.00,0.07,1.00}{wahid.faidi@lamsin.rnu.tn}
}
\begin{document}
\maketitle
\begin{abstract}
This article studies the problem of utility maximization in an incomplete market under a class of nonlinear expectations and general constraints on trading strategies.
Using a $g$-martingale method, we provide an explicit solution to our optimization problem for different utility functions and characterize an optimal investment-consumption strategy through the solutions to quadratic BSDEs.

\end{abstract}
\section{Introduction}
In financial theory, utility maximization serves as a foundational concept that guides the actions of investors in their pursuit of wealth. 
As stated by Von Neumann Morgestrern \cite{VM44}, classical (linear) expectations are the basis of traditional utility models, which assume that investors can average future gains without having to account for risk attitudes.
This representation has been challenged by later studies, namely the works of Allais \cite{Allais53}, Ellsberg \cite{Elsb61}, Marinacci \cite{Mar02} and Maccheroni et all.\cite{MMR06}. Such works has proven that in the presence of risks stemming principally from model uncertainty it is convenient to use non-linear expectations to represent investors' preferences.
In light of this obvious importance of the nonlinear representation of investors' preferences, Various alternative methods have been proposed to model the decision beyond the classical expected utility. Some examples include the concepts of capacity namely Choquet expectation introduced in  \cite{Choq53} , others use weighted expected utility, and more recently  robust utility,  recursive utility and $g$- expectation. Following this trend, we focus in this work on the portfolio optimization of an agent whose utility is represented by the solution of a nonlinear backward stochastic differential (BSDE), the $g$- expected utility. The notion of g-expectation is firsty introduced by Peng \cite{Peng97}
Since then, $g$- expectation theory has undergone considerable development. As in the case of classical expectation, a theory of $g$-martingales has developed over the past two decades. Some generalizations of the results concerning classical martingales have been made for $g$-martingales. Despite this considerable theoretical development of this notion, there are few treatises which deal with the $g$-expected utility optimal portfolio. This paper examines the optimal investment and consumption problem based on $g$- expected utility  and general stochastic constraints, not necessarily convex,  on consumption and investment strategies. We provide explicit solutions for investors with exponential, logarithmic, and power utilities in terms of solutions to BSDEs with quadratic growth generators. Our approach is based on an extension of the similar arguments of Hu et al. \cite{IH2005}, and Cheridito et al.\cite{cheridito10}. To each admissible strategy, we associate a utility process, which we show is always a $g$- supermartingale and a $g$- martingale if and only if the strategy is optimal. This method is based on Briand's results on the existence and properties of solutions to BSDEs with quadratic growth \cite{BriHu08}. We formulate constraints on consumption and investment in terms of subsets of predictable processes and use the conditional analysis results of Cheridito et al. \cite{cheridito11} to obtain the existence of optimal strategies. 
 This paper is organized as follows: The model is introduced in Section 2. The case of constant absolute risk aversion, which corresponds to exponential utility functions, is discussed in Section 3. Power and logarithmic utility are the subjects of Sections 4 and 5, respectively. Section 5 concludes with a discussion of the assumptions and potential generalizations.
\section{g-Expected Utility Maximization (g-EUM) model }\label{sec:model}
\subsection{g-expectation}
Let $T \in \mathbb{R}_+$ be a finite time horizon and $(W_t)_{0 \le t \le T}$
an $n$-dimensional Brownian motion on a probability
space $(\Omega, {\cal F}, \p)$. Denote by $({\cal F}_t)_{0 \le t \le T}$ the
augmented natural filtration of $(W_t)_{0 \le t \le T},$
$$\Fc_t = \sigma\{\sigma\{B_s; 0 \leq  s  \leq  t\} \cup \Nc \}, \Fc^0_{\infty}:=\bigcup\limits_{t>0}\Fc_t.$$
where $\Nc$ is the collection of $\P$–null sets in $\Omega$.
\\
The generator $g(t, \omega, y, z):[0, T] \times \Omega \times \mathbb{R} \times \mathbb{R}^d \longmapsto \R$ is a random function which is a progressively measurable stochastic process for any $(y, z)$. We assume that it satisfies the following  assumptions: 
\begin{itemize}
\item[$\mathbf{(A1)}$] $\forall (t,y,y',z,z')\in [0,T]\times \mathbb{R} \times \mathbb{R}\times \mathbb{R}^d \times \mathbb{R}^d;$
$$  \left|g(t, y, z)-g\left(t, y^{\prime}, z^{\prime}\right)\right| \leq \nu\left|y-y^{\prime}\right|+\phi(\left|z-z^{\prime}\right|),$$
where $\phi:\R_+\longrightarrow \R_+$, is subadditive and increasing with $\phi(0)=0$ and has a linear growth
with constant $\mu$, i.e., $\forall x\in \R_+, \phi(x)\leq \mu(x+1)$;
\item[$\mathbf{(A2)}$] $\forall y \in \mathbb{R} ;\;  g(t, y, 0)=0$,  $d\mathbb{P} \times dt$-a.e.
\end{itemize}
Note that under assumptions $\mathbf{(A1)}$ and $\mathbf{(A2)}$, we have forall $(y,z)\in \mathbb{R}\times \mathbb{R}^d$,
\begin{equation*}
\mathbb{E}\left[\left(\int_0^T|g(t, y,z)| dt\right)^2\right]=\mathbb{E}\left[\left(\int_0^T|g(t, y,z)-g(t,y,0)| dt\right)^2\right] \leq  \mathbb{E}\left[\left(\int_0^T|\phi^2(|z|) dt\right)^2\right]<+\infty,
\end{equation*}
and so, according to \cite{Jia10},  the BSDE 
\begin{equation}\label{bsde}
Y_t =\xi+\dint_t^{T}g\left(s, Y_s, Z_s\right) d s-\dint_t^{T}Z_s d B_s 
\end{equation}
admits a unique solution $(Y^{\xi},Z^{\xi}) \in \mathscr{S}^2(0, T ; \mathbb{R}) \times \mathscr{H}^2\left(0, T ; \mathbb{R}^d\right)$ for all $\xi \in  L^{2}\left(\Omega, \mathscr{F}_T, \P\right)$.
\\ The operator $\Ec_g$ defined by:
\begin{equation*}
\begin{split}
\Ec_g:L^{2}\left(\Omega, \mathscr{F}_T, \mathbb{P}\right)&\longmapsto \mathbb{R}
\\& \xi \longmapsto Y^{\xi}_0
\end{split}
\end{equation*}
 is a typical example of nonlinear expectation called $g$-expectation.
 \begin{Definition}
The system of operators
\begin{equation*}
\begin{split}
\Ec_g\left[\cdot \mid \mathcal{F}_t\right]: & L^2\left(\mathcal{F}_T\right) \rightarrow L^2\left(\mathcal{F}_t\right)\\ & \;\;\;\;\; \xi \longmapsto Y^{\xi}_t
\end{split}
\end{equation*}
 is called conditional  $g$-expectation   
 \end{Definition}
  If  $g = \phi(|z|)$ then we denote $\Ec_g$ by $\Ec^{\phi}.$ The conditional $g$-expectation is an $\mathcal{F}_t$-consistent nonlinear expectation i.e. it verifies the following properties:
 \begin{itemize}
\item[$\mathbf{(B1)}$] Monotonicity: $\mathcal{E}_g\left[Y \mid \mathcal{F}_t\right] \geq \mathcal{E}_g\left[Z \mid \mathcal{F}_t\right]$, a.s., if $Y \geq Z$, a.s.;
\item[$\mathbf{(B2)}$] Constant-preserving: $\mathcal{E}_g\left[Y \mid \mathcal{F}_t\right]=Y$, a.s., if $Y \in L^2\left(\mathcal{F}_t\right)$;
\item[$\mathbf{(B3)}$] Time consistency: $\mathcal{E}_g\left[\mathcal{E}_g\left[Y \mid \mathcal{F}_t\right] \mid \mathcal{F}_s\right]=\mathcal{E}_g\left[Y \mid \mathcal{F}_s\right]$, a.s., if $s \leq t \leq T$;
\item[$\mathbf{(B4)}$] "Zero-one law": for each $t, \mathcal{E}_g\left[1_A Y \mid \mathcal{F}_t\right]=1_A \mathcal{E}_g\left[Y \mid \mathcal{F}_t\right]$, a.s., $\forall A \in \mathcal{F}_t$.
\end{itemize}
 Reciprocally, if $\Ec$  is a conditional nonlinear expectation operator verifying $\mathbf{(B1-B4)}$ and the following domination assumption
\\ $\mathbf{(H1)}$:For each $X, Y$ in $L^2(\Fc_T)$, we have $\Ec[X|\Fc_t]- \Ec[Y|\Fc_t] \leq 
\Ec^{\phi}[X-Y|\Fc_t], \forall t\in [0,T]$, where $\phi$ is the function given in $\mathbf{(A1)}$,
then according \cite{Jia10}, there exists a function $g(t,z) : \Omega\times [0,T] \times \R^d\longrightarrow \R$  satisfying $\mathbf{(A1-A2)}$ such that $\Ec\equiv \Ec_g.$
\\ We now provide two basic instances of nonlinear expectation:
\begin{Example}{\textbf{The Standard (linear) expectation}} \\
Which correspond to linear generator 
$$g(t,y,z)=  \eta_t.z,$$ 
where $ \gamma: \Omega \times \mathbf{R}_{+} \rightarrow \mathbf{R}^n$  is bounded, adapted and continuous process. 
In this case $\Ec_g$ can be written 
$$
\Ec_g[\xi|\Fc_t]=\E_{\Q}\left(\xi  \mid \mathcal{F}_t\right)
$$
where the probability measure $\Q$ is given by
$$\dfrac{d\Q}{d\p}=\exp(\int_0^t\gamma_s dW_s -\frac{1}{2} \int_0^t\gamma^2_s ds ):=\Ec(\int_0^t\gamma_s dW_s).$$
\end{Example}
\begin{Example}{\textbf{The $\kappa$-ignorance expectation}} \\
Which correspond to  generator 
$$g(t,y,z)=  \kappa |z|,$$ where $\kappa > 0.$
In this case the functional $\Ec_g$ is represented as follows
$$\Ec_g[\xi |\Fc_t]= \essinf \limits_{Q^\theta \in \Pc^\theta}\E_{\Q}[\xi |\Fc_t],$$
$$\Pc^\theta=\{Q^\theta; \dfrac{dQ^\theta}{d\P}=\exp(-\int_0^.\theta_s dW_s -\dfrac{1}{2}\int_0^. |\theta_s|^2ds); |\theta_s| \leq \kappa\}$$
\end{Example}
\subsection{Market Model}
We examine a financial market that includes $m \le n$ stocks and a money market. At a fixed interest rate $r \ge 0$, money can be borrowed and lent from the money market, and stock prices fluctuate as follows:
$$
\frac{dS^i_t}{S^i_t} = \mu^i_t dt + \sigma^i_t dW_t, \quad S^i_0 > 0, \quad i =1 , \dots, m,
$$
Where $\mu^i_t$ and $\sigma^i_t$ are bounded predictable processes that take values in $\mathbb{R}$ and $\mathbb{R}^{1 \times n}$, respectively. Even in the absence of constraints, the market is incomplete if $m < n$, meaning that the stocks do not cover all uncertainty.

An investor with initial wealth $x \in \mathbb{R}$ can invest in the financial market and consume at intermediate times. They will receive a lump sum payment $E$ at time $T$ that is ${\cal F}_T$ measurable and income at a predictable rate $e_t$. If the investor invests using a predictable trading strategy $\pi_t$ and consumes at a predictable pace $c_t$, with values in $\mathbb{R}^{1 \times m}$, where $\pi^i_t$ is the amount invested in stock $i$ at time $t$, then his  wealth evolves as
$$
X_t = x+ \int_0^t \brak{X_s - \sum_{i=1}^m \pi^i_s} r ds +
\sum_{i=1}^m \int_0^t \frac{\pi^i_s}{S^i_s} dS^i_s + \int_0^t (e_s-c_s) ds.
$$
The matrix with rows $\sigma^i_t$, $i = 1, \dots, m$, is denoted as $\sigma_t$.
Assume that $\sigma \sigma^T$ is invertible $\nu \otimes \p$-almost everywhere, where
$\nu$ is the Lebesgue measure on $[0,T]$, and the process
$$
\theta = \sigma^T (\sigma \sigma^T)^{-1} (\mu - r1)
$$
is bounded by $\kappa$. From there, taking $p = \pi \sigma$, one can write
\be \label{XSDE}
X^{(c,p)}_t = x + \int_0^t X^{(c,p)}_s r ds +
\int_0^t p_s [dW_s + \theta_s ds] + \int_0^t (e_s-c_s) ds.
\ee
Note that if
$$
\int_0^T (|e_t| + |c_t| + |p_t|^2) dt < \infty \quad \p \mbox{-almost surely,}
$$
where $|.|$ denotes the Euclidean norm on $\mathbb{R}^{1 \times n}$, then
$$
\int_0^t p_t [dW_t + \theta_t dt] + \int_0^t (e_s-c_s) ds
$$
is a continuous stochastic process, and it follows that equation \eqref{XSDE} has a unique continuous
solution $X^{(c,p)}$ given by
\be \label{wealthprocess}
X^{(c,p)}_t = e^{rt}(x + \int_0^t e^{-rs} p_s [dW_s + \theta_s ds] + \int_0^t e^{-rs}(e_s-c_s) ds.
\ee

In the case of CRRA utility, we remark that is convenient to parametrize the investment-consumption strategies and the income process as fractions of wealth, so $p$, $c$ and $e$ are replaced by $pX^{(c,p)}$, $cX^{(c,p)}$ and $eX^{(c,p)}$ respectively. We will return to this remark in more detail in the corresponding paragraph.
Our financial agent seeks to find $c$ and $p$ (hence $\pi$) which allow them to maximize the following nonlinear expected utility:
\be \label{opt}
\Ec_g\Big[{\int_0^T \alpha e^{-\int_0^t \delta_sds} u(c_t) dt + \beta e^{-\int_0^T \delta_sds} u \brak{X^{(c,p)}_T -F}}\Big]
\ee
for given constants $\alpha , \beta > 0$, a deterministic bounded discounted factor  $ \delta$  and
a concave utility  function $u : \mathbb{R} \to \mathbb{R} \cup \crl{-\infty}$.
Along this work, we study the cases of the following classical utility functions:\\
\hspace*{5mm} $\bullet$ $u(x) = - \exp(-\gamma x)$ for $\gamma > 0$\\
\hspace*{5mm} $\bullet$ $u(x) = \log(x)$\\
\hspace*{5mm} $\bullet$ $u(x) = x^{\gamma}/\gamma$ for $\gamma \in (-\infty,0) \cup (0,1)$.\\
To solve this problem, we adopt the same method as Imkeller and Hu \cite{IH2005}, Cheridito et all \cite{cheridito10}  which consists of constructing a processes family $R^{(p,c)}$  verifying
\begin{itemize}
\item  $R_T^{(\pi,c)}=\int_0^T \alpha e^{-\int_0^t \delta_sds} u(c_t) dt + \beta e^{-\int_0^T \delta_sds} u \brak{X^{(c,p)}_T -F}$ for all $(p,c)$,
\item $R_0^{(\pi,c)}=R_0$ is constant for all  $(p,c)$,
\item  for all $(p,c)$,  $R^{(p,c)}$ is a $g$-supermartingale  and there exists a $(p^*,c^*)$ such that $R^{(p^*,c^*)}$ is a $g$-martingale.
\end{itemize}
Thus, 
$$\Ec_g(R_T^{(p,c)}) \leq R_0=\Ec_g(R_T^{(p^*,c^*)}),$$
which means that $(p^*,c^*)$ is optimal.
To formulate consumption and investment constraints we introduce non-empty subsets
$\Cc \subset {\cal P}$ and $\Qc \subset {\cal P}^{1 \times m}$, where ${\cal P}$ denotes the
set of all real-valued predictable processes $(c_t)_{0 \le t \le T}$ and
${\cal P}^{1 \times m}$ the set of all predictable processes $(\pi_t)_{0 \le t \le T}$ with values in
$\mathbb{R}^{1 \times m}$. In Section \ref{sec:exp} we do not put restrictions on
consumption and just require the investment strategy $\pi$ to belong to $Q$. In Sections \ref{sec:power} and \ref{sec:log}
consumption and investment will be of the form $c = \tilde{c} X$
and $\pi = \tilde{\pi} X$, respectively, and we will require $\tilde{c}$ to be in $C$ and
$\tilde{\pi}$ in $Q$.

Note that the expected value \eqref{opt} does not change if $(c,p)$ is
replaced by a pair $(c',p')$ which is equal $\nu \otimes \p$-a.e.
So we identify predictable processes that agree
$\nu \otimes \p$-a.e. and use the following concepts from Cheridito et al. \cite{cheridito11}:
We call a subset $A$ of ${\cal P}^{1 \times k}$ {\bf sequentially closed} if
it contains every process $a$ that is the $\nu \otimes \p$-a.e. limit
of a sequence $(a^n)_{n \ge 1}$ of processes in $A$. We call it
${\cal P}$-{\bf stable} if it contains $1_B a + 1_{B^c} a'$ for all $a,a' \in A$ and
every predictable set $B \subset [0,T] \times \Omega$. We say $A$ is
${\cal P}$-{\bf convex} if it contains $\lambda a + (1-\lambda) a'$
for all $a,a' \in A$ and every process $\lambda \in {\cal P}$ with values in $[0,1]$.
In the whole paper we work with the following\\[3mm]
\hspace*{7mm}{\bf Standing assumption} \quad $\Cc$ and $\Qc$ are sequentially closed and ${\cal P}$-stable.\\[3mm]
This will allow us to show existence of optimal strategies. If, in addition,
$C$ and $Q$ are ${\cal P}$-convex, the optimal strategies will be unique.
Note that $P = \crl{\pi \sigma : \pi \in Q}$ is a ${\cal P}$-stable subset
of ${\cal P}^{1 \times n}$, which, since we
assumed $\sigma \sigma^T$ to be invertible for $\nu \otimes \p$-almost all $(t,\omega)$,
is ${\cal P}$-convex if and only if $Q$ is. Moreover, it follows from \cite{cheridito11}
that $P$ is  sequentially closed.
In the sequel, we assume also that $g$ is positively homogeneous in $z$ that's mean
\\
$\mathbf{H2}:\forall (t,z)\in \R \times \R^d,\forall \lambda \geq 0 ; \;\; g(t, \lambda z)= \lambda g(t,z) \P.a.s.$

\section{Exponential Utility}\label{sec:exp}
In this case $
u(x) = - \exp(-\gamma x),
$ which correspond that  investor has constant absolute risk aversion
$- u''(x)/u'(x) = \gamma > 0$. 
We assume that the set $\Acr$ of possible investment strategies contains at least one
bounded process $\bar{p}$. Moreover, we assume that the rate of income
$e$ and the final payment $E$ are both bounded.

Define the bounded positive function $h$ on $[0,T]$ by
$$
h(t) = 1/(1 + T -t) \quad \mbox{if} \quad r = 0
$$
and
$$
h(t) = \frac{r}{1 -(1-r) \exp(-r (T-t))}  \quad \mbox{if} \quad r > 0.
$$
Note that in both cases $h$ solves the quadratic ODE
$$
h'(t) = h(t)(h(t) -r), \quad h(T) = 1.
$$

\begin{Definition}
If $u(x) = - \exp(-\gamma x)$, an admissible strategy consists of 
$(p,c) \in {\cal P}\times {\cal C}$
such that  $\dint_0^T|c_t|dt < \infty \p.a.s$ ,  $\E[(\dint_0^T e^{-\gamma c_t}dt)^2]< \infty$  and
$$
\exists q > 2 \quad  \E\Big[\brak{\int_0^T|p_t|^2}^{\frac{q}{q-2}}+\sup\limits_{0\leq t \leq T}\exp\brak{- \gamma q h(t) X^{p,c}_t}\Big] < \infty  .
$$
\end{Definition}
For this study we begin by the following useful lemma
\begin{Lemma}\label{lemma1}
$\forall Z\in \Pc; \;\; \argmin\limits_{p\in \Ac}(\frac{\gamma}{2}|hp-Z|^2+hp\theta - g(hp-Z))$ is nonempty subset from $\Pc.$ Moreover $\forall p\in \argmin\limits_{p\in \Ac}(\frac{\gamma}{2}|hp-Z|^2+hp\theta - g(hp-Z))$, there are two constants $\kappa_1, \kappa_2$ such that 
\begin{equation}\label{estimatep}
|p|^2\leq \kappa_1|Z|^2+\kappa_2\;\; dt\otimes d\p. a.e .
\end{equation}
\end{Lemma}

\begin{proof}
Let $f_Z:\Ac \longrightarrow \bar{L}; p\mapsto \frac{\gamma}{2}|hp-Z|^2+hp\theta - g(hp-Z).$
\\$f_Z$ is sequentially continuous stable and the set 
$$\Gamma:=\{p\in \Ac\;\;s.t\;\; f(p)\leq f(\bar{p})\}$$
is $L^0-$bounded, indeed, let $p \in \Gamma$ we have
$$
\begin{aligned}
\frac{\gamma}{2}|hp-Z|^2+hp\theta - g(hp-Z) \leq \frac{\gamma}{2}|h\bar{p}-Z|^2+h\bar{p}\theta - g(h\bar{p}-Z)
\end{aligned}
$$
Using Cauchy Schwartz inequality and assumption on $g$, we obtain
$$
\begin{aligned}
\frac{\gamma}{2}|hp-Z|^2- |hp-Z||\theta| - \mu(|hp-Z|+1) \leq \frac{\gamma}{2}|h\bar{p}-Z|^2+|h\bar{p}-Z||\theta| + \mu(|h\bar{p}-Z|+1)
\end{aligned}
$$
this implies, for all non negative real $\alpha$
$$
\begin{aligned}
\frac{\gamma}{2}|hp-Z|^2- \frac{\kappa+\mu}{2}\alpha^2 |hp-Z|^2-\frac{\kappa+\mu}{2\alpha^2} - \mu \leq \frac{\gamma}{2}|h\bar{p}-Z|^2+\frac{\kappa +\mu}{2}|h\bar{p}-Z|^2 + \frac{\kappa +\mu}{2} + \mu
\end{aligned}
$$
By choosing $\alpha$  small enough so that $\delta:=\frac{\gamma}{2}-\frac{\kappa+\mu}{2}\alpha^2$ is positive, there are two constants $a$ and $b$ such that 
$$|hp-Z|^2 \leq a|h\bar{p}-Z|^2+b$$ 
Thus, 
\begin{equation}\label{estimatep2}
|p|^2 \leq \frac{2}{h}(|hp-Z|^2+|Z|^2)\leq \frac{2}{h}(|h\bar{p}-Z|^2+|Z|^2)
\end{equation}
$$$$
So $\Gamma$ is $L^0-$bounded. Consequently, using Theorem 4.4 in \cite{cheridito11}, there exists at least $\hat{p} \in \Ac$ such that $\essinf\limits_{p\in \Ac}f(p)=f(\hat{p}).$
\\
Estimate \ref{estimatep} is an immediate consequence of the inequality \ref{estimatep2} and  the fact that $h$ and $\bar{p}$ are bounded.
\end{proof}

Consider the BSDE
\be \label{BSDEexp}
Y_t = F + \int_t^T f(s,Y_s,Z_s)ds - \int_t^T Z_s dW_s
\ee
with driver
$$
f(t,y,z)
= h(t)(e_t-y)+ \esssup\limits_{p\in {\Acr} }[ g(t, h(t)p-z)+ h(t)p\theta_t -\frac{\gamma}{2} |h(t)p-z|^2]+\frac{h(t)}{\gamma} \brak{\log \frac{h(t)}{\alpha} - 1}
+ \frac{\delta_t}{\gamma}.
$$
where ess\,sup denotes the smallest upper bound with respect to the $\nu \otimes \p$-a.e. order.
Since $\theta$, $e$, $E$ and $h$ are bounded, using Lemma \ref{lemma1}, there exists a positive constant $K $ such that
$$|f(t,y,z)| \le K(1+ |y| + |z|^2)$$
and
$$|f(t,y_1,z_1)-f(t,y_2,z_2)|\le K( |y_1-y_2|+(1+|z_1|+|z_2|)|z_1-z_2|).$$
So it follows from Kobylanski \cite{Kobylanski} that equation \eqref{BSDEexp} has a unique solution
$(Y,Z)$ such that $Y$ is bounded and from  Morlais \cite{Morlais} that $Z$
belongs to ${\cal P}^{1 \times n}_{\rm BMO}$.
According Briand and Hu \cite{BH08}, $(Y,Z)$ satisfies for each $p > 1$
\begin{equation}\label{YZestimation}
\mathbb{E}\left[\exp \left(\gamma p \sup _{0 \leq t \leq T}\left|Y_t\right|\right)+\left(\int_0^T\left|Z_s\right|^2 d s\right)^{p / 2}\right] \leq C \mathbb{E}\left[\exp \left(2p K\left(|F|+KT\right)\right)\right]. 
\end{equation}

\begin{Theorem} \label{thmexp}
The optimal value of the optimization problem \eqref{opt} for
$u(x) = - \exp(-\alpha x)$ over all admissible strategies is
\be \label{optvalueexp}
v(0)=- \exp \edg{- \gamma (h(0) x + Y_0)},
\ee
and there is at least one optimal strategy $(p^*,c^*)$ given by
\be \label{optstexpp}
p^* \in \argmax{g(t, h(t)p-Z)- h(t)p\theta_t -\frac{\gamma}{2} |h(t)p-Z|^2} 
\ee
and 
\beas \label{optstexpc}
 c^*_t = && h(0)x+\int_0^t h(s)(p_s^*\theta_s+e_s-Y_s+\frac{1}{\gamma}\ln(\frac{h(s)}{\gamma}))ds+\int_0^t h(s)p^*_sdW_s + Y_t - \frac{1}{\gamma} \log \frac{h(t)}{\alpha}.
\eeas
\end{Theorem}

\begin{proof}
For every admissible strategy $(c,p)$, we consider the continuous stochastic process
$$
R^{(p,c)}_t = -\dint_0^t \alpha e^{-\int_0^u \delta_sds} e^{-\gamma c_u} du - e^{-\int_0^t \delta_sds} \exp\brak{-\gamma (h(t)X^{(p,c)}_t -Y_t)}.
$$
We have
\begin{itemize}
\item $\forall (p,c) \in \Acr \times \Ccr; \E[\sup\limits_{0\leq t \leq T}(R^{(p,c)}_t)^2] < \infty$
\item $\forall (p,c) \in \Acr \times \Ccr; R^{(p,c)}_T = -\dint_0^T \alpha e^{-\int_0^u \delta_sds} e^{-\gamma c_u} du - e^{-\int_0^T \delta_sds} \exp\brak{-\gamma (h(t)X^{(p,c)}_T -F)},$
\item $\forall t\in [0,T]$,  $\forall (p,c), (p',c') \in \Acr \times \Ccr $ such that $(p,c)_., (p',c')_.$ on $[0,t]$, we have  $$R_t^{(p,c)}=R_t^{(p',c')}, $$
\end{itemize}
Using it\^o formula, we obtain
$$
dR^{(p,c)}_t =  \gamma e^{- \int_0^t \delta_u du} e^{- \gamma \brak{h(t) X^{(p,c)}_t - Y_t}}
\edg{-g(t,h(t) p_t - Z_t)+ (h(t) p_t - Z_t) dW_t +  A^{(p,c)}_t dt},
$$
$$
dR^{(p,c)}_t = \edg{-g(Z_t^{p,c})+ Z_t^{p,c} dW_t +  A^{(p,c)}_t dt},
$$
where
\beas
Z_t^{p,c}=\gamma e^{- \int_0^t \delta_u du} e^{- \gamma \brak{h(t) X^{(p,c)}_t - Y_t}}(h(t) p_t - Z_t),
\eeas
and
\beas
A^{(p,c)}_t = &&\gamma e^{- \int_0^t \delta_u du} e^{- \gamma \brak{h(t) X^{(p,c)}_t - Y_t}}[g(t,h(t) p_t - Z_t)+ h(t) p_t \theta_t - \frac{\gamma}{2} |h(t) p_t - Z_t|^2 - f(t,Y_t,Z_t)\\
&& + h(t) (e_t - c_t) - \frac{\alpha}{\gamma} e^{\gamma \brak{h(t) X^{(p,c)}_t + Y_t}} e^{-\gamma c_t}
+ h'(t) X^{(p,c)}_t + h(t) r X^{(p,c)}_t + \frac{\delta_t}{\gamma}]  \\
=&&\gamma e^{- \int_0^t \delta_u du} e^{- \gamma \brak{h(t) X^{(p,c)}_t - Y_t}}[ g(t,h(t) p_t - Z_t)+ h(t) p_t \theta_t - \frac{\gamma}{2} |h(t) p_t - Z_t|^2 - f(t,Y_t,Z_t)\\
&& + h(t) (e_t - c_t) - \frac{\alpha}{\gamma} e^{\gamma \brak{h(t) X^{(p,c)}_t + Y_t}} e^{-\gamma c_t}
+\frac{\delta_t}{\gamma} ]
\eeas
let's check that process $A^{(p,c)}$ satisfies
$\E[\int_0^T(A^{(p,c)}_s)^2ds] < \infty$ for each $t > 0.$ 
\\For all admissible strategy $(p,c)$, we have
\beas
\E[\dint_0^T |Z_t^{(p,c)}|^2dt]&=\E[\dint_0^T \gamma^2 e^{-2 \int_0^t \delta_u du} e^{- 2\gamma \brak{h(t) X^{(p,c)}_t - Y_t}}|(h(t) p_t - Z_t)|^2dt] \\
& \leq K \E[\sup\limits_{0 \leq t \leq T} e^{- \gamma h(t) X^{(p,c)}_t } \dint_0^T |(h(t) p_t - Z_t)|^2dt]\\
& \leq K \E[\sup\limits_{0 \leq t \leq T} e^{- q\gamma h(t) X^{(p,c)}_t }]^{\frac{2}{q}} \E[\dint_0^T |(h(t) p_t - Z_t)|^2dt)^{\frac{q}{q-2}}]^{\frac{q-2}{q}}<\infty.
\eeas
Where the first inequality come from the boundness of $\delta$ and $Y$ and the second is the Holder inequality. Which implies that
$\int_0^T Z_t^{(p,c)} dW_t  $
and $\int_0^T g(Z_t^{(p,c)} )t  $ are both in $L^2(\Omega, \Fc,\p).$
So, for all admissible strategy $(p,c),$ we have $\int_0^T A^{(p,c)}_t dt \in L^2(\Omega, \Fc,\p).$
\\On the other hand, for fixed $(t,\omega) \in [0,T] \times \Omega$, function of the real variable 
$$
z \mapsto - h(t) z - \frac{\alpha}{\gamma} e^{\gamma \brak{h(t) X^{(c,p)}_t + Y_t}} e^{-\gamma z}
$$
is a strictly concave function that is equal to its maximum
$$
\frac{h(t)}{\gamma} \log \frac{h(t)}{\alpha} - h^2(t) X^{(c,p)}_t - h(t) Y_t
- \frac{h(t)}{\gamma}
$$
if and only if
$$
z = h(t) X^{(c,p)}_t + Y_t - \frac{1}{\gamma} \log \frac{h(t)}{\alpha}.
$$
Therefore, one has
\bea
\notag
&& h(t) (e_t - c_t) - \frac{\alpha}{\gamma} e^{\gamma \brak{h(t) X^{(c,p)}_t + Y_t}} e^{-\gamma c_t}
+ h'(t) X^{(c,p)}_t + h(t) r X^{(c,p)}_t + \frac{\delta_t}{\gamma}\\
\notag
&\le& h(t) e_t + \frac{h(t)}{\gamma} \log \frac{h(t)}{\alpha} - h^2(t) X^{(c,p)}_t - h(t) Y_t
- \frac{h(t)}{\gamma}
+ h'(t) X^{(c,p)}_t + h(t) r X^{(c,p)}_t + \frac{\delta_t}{\gamma}\\
\label{disX}
&=& h(t) e_t + \frac{h(t)}{\gamma} \log \frac{h(t)}{\alpha} - h(t) Y_t - \frac{h(t)}{\gamma}
+ \frac{\delta_t}{\gamma},
\eea

It follows that, for all admissible strategy $(p,c)$,
\beas
A^{(p,c)}_t \leq  &&\gamma e^{- \int_0^t \delta_u du} e^{- \gamma \brak{h(t) X^{(p,c)}_t - Y_t}}[ g(t,h(t) p_t - Z_t)+ h(t) p_t \theta_t - \frac{\gamma}{2} |h(t) p_t - Z_t|^2 - f(t,Y_t,Z_t)\\
&& + h(t) e_t + \frac{h(t)}{\gamma} \log \frac{h(t)}{\alpha} - h(t) Y_t - \frac{h(t)}{\gamma}
+ \frac{\delta_t}{\gamma}]  \\
\leq && 0 \;\;\;  \mathbb{P}\text{.a.s}
\eeas
So, for all $(p,c) \in \Acr$, $R^{(p,c)}$ is a $g-$ supermartingale.
It remains to be verified that there is $(p^*,c^*)\in \Acr $ such  that $R^{(p^*,c^*)}$ is a $g-$ martingale,  for this we must ensure that $A^{(p^*,c^*)}_t=0 \;\;dt\times d\P .a.e.$  \\
Let  $ p^* \in \argmax{g(t, h(t)p-Z)- h(t)p\theta_t -\frac{\gamma}{2} |h(t)p-Z|^2},$ we have 
\beas
&&g(t,h(t) p^*_t - Z_t)+ h(t) p^*_t \theta_t - \frac{\gamma}{2} |h(t) p^*_t - Z_t|^2 - f(t,Y_t,Z_t)\\
&& + h(t) (e_t-Y_t +\frac{h(t)}{\gamma}( \log \frac{h(t)}{\alpha}-1)+ \frac{\delta_t}{\gamma}  =0 \;\;dt\times d\P .a.e.
\eeas
On the other hand, the linear  SDE
$$dX_t=(r-h(t))X_tdt + (p^*_t\theta_t + e_t-Y_t+\frac{1}{\gamma}\ln (\dfrac{h(t)}{\alpha}))dt+p^*_tdW_t, X_0=x $$ 
has a unique solution given by
$$X_t=\dfrac{1}{h(t)}(h(0)x+\int_0^th(s)(p^*_s\theta_s + e_s -Y_s+\frac{1}{\gamma}\ln (\dfrac{h(s)}{\alpha}))ds +\int_0^t h(s)p^*_sdW_s).$$
We have , for all $q \geq 2$, 
\beas
e^{-q\gamma  h(t) X_t}&& = e^{\displaystyle{-q\gamma  h(0)x -q\gamma  \int_0^t h(s)p^*_s \theta_s ds -q\gamma  \int_0^t h(s)p^*_sdW_s -q\gamma  \int_0^t h(s)( e_s -Y_s+\frac{1}{\gamma}\ln (\frac{h(s)}{\alpha}))ds } }\\
&& \leq d_1  e^{\displaystyle{ -q\gamma  \int_0^t h(s)p^*_s \theta_s ds -q\gamma  \int_0^t h(s)p^*_sdW_s } } \\
&& \leq d_1 e^{\displaystyle{  (-q\gamma  \int_0^t h(s)p^*_sdW_s-q\gamma^2  \int_0^t h^2(s)|p^*|^2_s \theta_s ds )} }  e^{\displaystyle{ (-q\gamma  \int_0^t h(s)p^*_s \theta_s ds +q\gamma^2  \int_0^t h^2(s)|p^*|^2_s \theta_s ds) } } \\
&& \leq d_2 \Big(e^{\displaystyle{  (-2q\gamma  \int_0^t h(s)p^*_sdW_s-2q\gamma^2  \int_0^t h^2(s)|p^*|^2_s \theta_s ds )} }
\\&& +  e^{\displaystyle{ (-2q\gamma  \int_0^t h(s)p^*_s \theta_s ds +2q\gamma  \int_0^t h^2(s)|p^*|^2_s \theta_s ds) } } \Big)\\
&&\leq d_3 ((\Ec(-2\gamma  \int_0^t h(s)p^*_sdW_s))^q +  \exp( d_4\int_0^T |p^*|^2_s  ds)),
\eeas
where $d_1,d_2,d_3$ and $d_4$ are positive constants. So, 
\begin{equation}\label{maxestimate}
\sup\limits_{0\leq t \leq T} e^{-\gamma  h(t) X_t} \leq d_3 (\sup\limits_{0\leq t \leq T}\Ec(-2\gamma  \int_0^t h(s)p^*_sdW_s)^q + \exp( d_4 \int_0^T |p^*|^2_s  ds))\;\;\; \p.a.s
\end{equation}
\\ Note that for all 
\begin{equation}\label{novikovcondition}
\forall\lambda>0,\quad\quad \E[\exp{(\lambda\dint_0^T |p^*|^2_s  ds)}]< \infty.
\end{equation}
Indeed, using Lemma \ref{lemma1}, convex inequality  and estimation (\ref{YZestimation}), there two positive constants $C_1$ and $M_1$ such that  
$$\forall k \in \mathbb{N} \setminus \{0\},\;\; \mathbb{E}\left[\left(\lambda \int_0^T\left|p^*_s\right|^2 d s\right)^{k}\right]<C_1M_1^k.$$
Which implies that
$$\forall n \in \mathbb{N} \setminus \{0\},\;\; \mathbb{E}\left[\sum\limits_{k=1}^n\frac{1}{k!}\left(\lambda \int_0^T\left|p^*_s\right|^2 d s\right)^{k}\right]<{C_1}(e^{M_1}-1).$$
By Fatou lemma, we deduce that 
$$\E[\exp{(\lambda\int_0^T |p^*|^2_s  ds)}]<C_1(e^{M_1}-1)+1< \infty.$$
The continuous process $\brak{-2\gamma  \dint_0^t h(s)p^*_sdW_s}_{0\leq t \leq T}$  is a BMO martingale then the process $M_.=(\Ec(-2\gamma \dint_0^. h(s)p^*_sdW_s))_{0\leq t \leq T}$ is uniformly integrable martingale. By Doob’s martingale maximal inequalities, we have, forall $q>1$
$$\E[(\sup\limits_{0\leq t \leq T}|M_t|)^q] \leq (\dfrac{q}{q-1})^q\E(|M_T|^q).$$
\beas
 (M_T)^{q}&& = e^{\displaystyle{-2q\gamma \dint_0^T h(s)p^*_sdW_s-2q\gamma^2 \int_0^Th^2(s)|p_s^*|^2ds}} \\
&&= e^{\displaystyle{-2q\gamma \int_0^Th(s) p_s^*dW_s-4q^2\gamma^2 \int_0^Th^2(s)|p_s^*|^2ds}}\times e^{\displaystyle{2q(2q-1)\gamma^2 \int_0^Th^2(s)|p_s^*|^2ds}} \\
&& \leq \dfrac{1}{2}\brak{e^{\displaystyle{-4q\gamma \int_0^Th(s) p_s^*dW_s-8q^2\gamma^2 \int_0^Th^2(s)|p_s^*|^2ds}}+ e^{\displaystyle{4q(2q-1)\gamma^2 \int_0^Th^2(s)|p_s^*|^2ds}}}
\eeas
In the previous inequality, the first term is a true martingale therefore of finite expectation, the second term is also of finite expectation taking into account (\ref{novikovcondition}).

Which means that $\E[\sup\limits_{0\leq t \leq T} e^{-\gamma q  h(t) X_t}] < \infty.$
By choosing 
\beas
c_t^*&& =h(t)X_t+Y_t-\frac{1}{\gamma}\ln (\dfrac{h(t)}{\alpha})\\
&& =h(0)x+\int_0^th(s)(p^*_s\theta_s + e_s -Y_s+\frac{1}{\gamma}\ln (\dfrac{h(s)}{\alpha}))ds +\int_0^t h(s)p^*_sdW_s+Y_t-\frac{1}{\gamma}\ln (\dfrac{h(t)}{\alpha}),
\eeas
Since $Y_t-\frac{1}{\gamma}\ln (\dfrac{h(t)}{\alpha})$ is bounded, there is a positive constant $d_5$ such that
$$(\int_0^Te^{-\gamma c_t^*})^2 \leq \int_0^Te^{-2\gamma c_t^*} \leq d_5 \sup\limits_{0\leq t \leq T} e^{- 2 \gamma   h(t) X_t}$$
and so $\E[(\int_0^Te^{-\gamma c_t^*})^2]< \infty.$
we have $X=X^{(p^*,c^*)}$ and so,
\begin{equation*} 
h(t) X^{(c^*,p^*)}_t + Y_t - \frac{1}{\gamma} \log \frac{h(t)}{\alpha}=c^*_t\;\; dt\times d\P .a.e.
\end{equation*}
It follows that,
$$
- h(t) c^*_t - \frac{\alpha}{\gamma} e^{\gamma \brak{h(t) X^{(c^*,p^*)}_t + Y_t}} e^{-\gamma c^*_t}
=
\frac{h(t)}{\gamma} \log \frac{h(t)}{\alpha} - h^2(t) X^{(c^*,p^*)}_t - h(t) Y_t
- \frac{h(t)}{\gamma}
$$
Consequently,
\beas
A^{(p^*,c^*)}_t = &&\gamma e^{- \int_0^t \delta_u du} e^{- \gamma \brak{h(t) X^{(p^*,c^*)}_t - Y_t}}[g(t,h(t) p^*_t - Z_t)+ h(t) p^*_t \theta_t - \frac{\gamma}{2} |h(t) p^*_t - Z_t|^2 - f(t,Y_t,Z_t)\\
&& + h(t) (e_t - c^*_t) - \frac{\alpha}{\gamma} e^{\gamma \brak{h(t) X^{(p^*,c^*)}_t + Y_t}} e^{-\gamma c^*_t}
+ h'(t) X^{(p^*,c^*)}_t + h(t) r X^{(p^*,c^*)}_t + \frac{\delta_t}{\gamma}]  \\
=&&\gamma e^{- \int_0^t \delta_u du} e^{- \gamma \brak{h(t) X^{(p^*,c^*)}_t - Y_t}}[ g(t,h(t) p^*_t - Z_t)+ h(t) p^*_t \theta_t - \frac{\gamma}{2} |h(t) p^*_t - Z_t|^2 - f(t,Y_t,Z_t)\\
&& + h(t) e_t +\frac{h(t)}{\gamma} \log \frac{h(t)}{\alpha} - h^2(t) X^{(c^*,p^*)}_t - h(t) Y_t
- \frac{h(t)}{\gamma} 
+ h'(t) X^{(p^*,c^*)}_t + h(t) r X^{(p^*,c^*)}_t + \frac{\delta_t}{\gamma} ] \\
=&&\gamma e^{- \int_0^t \delta_u du} e^{- \gamma \brak{h(t) X^{(p^*,c^*)}_t - Y_t}}[ g(t,h(t) p^*_t - Z_t)+ h(t) p^*_t \theta_t - \frac{\gamma}{2} |h(t) p^*_t - Z_t|^2 - f(t,Y_t,Z_t)\\
&& + h(t) (e_t-Y_t +\frac{h(t)}{\gamma}( \log \frac{h(t)}{\alpha}-1)+ \frac{\delta_t}{\gamma}]  =0 \;\;dt\times d\P .a.e.
\eeas

\end{proof}
\begin{Remark}\label{remexp}
In cheridito and all the proof of theorem \ref{thmexp} is based on the BMO character of the process $Z$. In ours situation, estimation (\ref{YZestimation}) plays a fundamental role in the proof. Thanks to this estimation, a slight modification in the set of admissible strategies allows us to assume that $e$ and $F$ are not bounded.
\end{Remark}
\begin{Example}
In the case of linear generator $g(t,z)=\eta_t.z$, the driver $f$ is given by
$$
f(t,y,z)
= h(t)(e_t-y)-\frac{\gamma}{2}{\rm dist}^2_t(z+\frac{\theta-\eta}{\gamma}, hP)+z\theta_t+\dfrac{1}{2\gamma}|\theta_t-\eta_t|^2+\frac{h(t)}{\gamma} \brak{\log \frac{h(t)}{\alpha} - 1}
+ \frac{\delta_t}{\gamma}.
$$
So, by setting $\eta\equiv 0$, we find the results of Cheridito et all \cite{cheridito11}.
If there is no constraint, we have ${\rm dist}_t(z+\frac{\theta-\eta}{\gamma}, hP)=0$, and so BSDE(\ref{BSDEexp}) is a linear BSDE whose solution is given by

$$
Y_t=\E\left[ \Gamma_T^t F+\int_t^T \Gamma_s^t \varphi_s d s \mid \mathcal{F}_t\right],
$$

where for $s \geq t\;\;\; \Gamma_s^t$ is given by

$$
\Gamma_s^t=\exp(\int_t^s(h(s)-\frac{1}{2}\theta_s^2)ds +\int_t^s\theta_sdW_s)
$$
and
$$
\varphi_t= -h(t)e_t-\dfrac{1}{2\gamma}|\theta_t-\eta_t|^2-\frac{h(t)}{\gamma} \brak{\log \frac{h(t)}{\alpha} - 1}
+ \frac{\delta_t}{\gamma}.
$$
The optimal investment strategy $p^*$ is equal to $Z+\frac{\theta-\eta}{\gamma}$ and the optimal value is given by
$$v(0)=- \exp \edg{- \gamma (h(0) x + \E[ \Gamma_T^0 F +\int_0^T \Gamma_s^0 \varphi_s d s ])}.$$
\end{Example}
\begin{Example}{The $\kappa$- ignorance utility}\\
We consider the one-dimensional context, and we assume there is no constraint on trading strategy ($\Qc=\Pc$). By studing the functionnal
$$p\longmapsto \kappa | h(t)p-z) | + h(t)p\theta_t -\frac{\gamma}{2} |h(t)p-z|^2$$
We obtain that 
$$f(t,y,z)= h(t)(e_t-y)+z+ \kappa^2+ \theta_t^2+ 2\kappa|\theta_t|+\frac{h(t)}{\gamma} \brak{\log \frac{h(t)}{\alpha} - 1}
+ \frac{\delta_t}{\gamma}.$$
$$p^*=(\dfrac{Z}{h}+\dfrac{\kappa+\theta}{h})\textbf{1}_{\theta>0}+(\dfrac{Z}{h}+\dfrac{-\kappa+\theta}{h})\textbf{1}_{\theta<0}.$$
Yet again, we come back to a linear backward stochastic equation. So the optimal value is ginven by
$$v(0)=- \exp \edg{- \gamma (h(0) x + \E[ \Gamma_T^0 F +\int_0^T \Gamma_s^0 \varphi_s d s ])}.$$
where for all $t \geq 0$ and for all $s \geq t\;\;\;$

$$
\Gamma_s^t=\exp(W_s-W_t-\frac{1}{2} (s-t)+\int_t^sh(s)ds)
$$
and
$$
\varphi_t= -h(t)e_t-\kappa^2- \theta_t^2- 2\kappa|\theta_t|-\frac{h(t)}{\gamma} \brak{\log \frac{h(t)}{\alpha} - 1}
+ \frac{\delta_t}{\gamma}.
$$
\end{Example}

\section{Power utility}\label{sec:power}
\label{subsec:power}
In this section, we are interested to the utility maximization problem with respect to

$$
u_\gamma(x)=\frac{1}{\gamma} x^\gamma, \quad x \geq 0, \quad \gamma \in (-\infty,0) \cup (0,1) .
$$

This time, our investor does not incur any  liabilities and maximizes the predicted utility of his wealth at time $T$. 
Therefore, we can parameterize $e$, $c$ and $\pi$ by
$\tilde{e} = e/X$, $\tilde{c} = c/X$ and $\tilde{\pi} = \pi/X$, respectively.
Starting from strictly positive initial wealth $x$, the investor process wealth  $X^{(\tilde{\pi},\tilde{c})}$  evolves according to
$$
\frac{d X^{(c,p)}_t}{X^{(c,p)}_t} = \tilde{p}_t(dW_t + \theta_t dt)
+ (r + \tilde{e}_t - \tilde{c}_t) dt, \quad X^{(c,p)}_0 = x,
$$
where $\tilde{p} = \tilde{\pi} \sigma$.
and one can write
\be \label{posX}
X^{(c,p)}_t =
x \, {\cal E} \brak{\tilde{p} \cdot W^{\q}}_t \exp\brak{\int_0^t( r
+ \tilde{e}_s - \tilde{c}_s) ds} > 0,
\ee
where ${\cal E}$ is the stochastic exponential and $W^{\q}_t = W_t +\int_0^t \theta_sds$. 
\\
The trading strategies $(\tilde{\pi}, \tilde{c})$ are constrained to take values in a sequentially closed and ${\cal P}$-stable set  $\Qc \times \Cc  \subseteq {\cal P}^{1 \times m}\times {\cal P}$. 
We make the following assumption:
\be \label{asspower}
\mbox{there exists a pair }
(\bar{p},\bar{c}) \in \Qc \times \Cc \mbox{ such that } \bar{p} \mbox{ and }
u_{\gamma}(\bar{c}) - \bar{c}   \mbox{ are bounded.}
\ee
Note that this implies that $u_{\gamma}(\bar{c})$ and $\bar{c}$ are both bounded.
\begin{Definition}
An admissible strategy consists of 
$(\tilde{\pi}, \tilde{c}) \in {\cal P}\times {\cal C}$
such that

$$
\tilde{c} \ge 0 \quad \nu \otimes \p \mbox{-a.e.} \quad \mbox{and} \quad\exists q > 2\;\; \text{such that}\;\; \E[(\int_0^T c^\gamma_t dt)^{\frac{q}{q-1}}+(\dint_0^T |\tilde{p}_t|^2dt)^{\frac{q}{q-2}}+ \sup\limits_{0\leq t \leq T} (X^{(p,c)}_t)^{q\gamma}] < \infty  .
$$
 \end{Definition}

Now consider the BSDE
\be \label{BSDEpower}
Y_t =  \int_t^T f(s,Y_s,Z_s)ds - \int_t^T Z_s dW_s
\ee
with driver
\begin{equation} \label{driverpower}
f(t,y,z)
= \gamma \brak{\essinf\limits_{p\in \Pc}\brak{\frac{1-\gamma}{2}|p_t|^2-g(p-\frac{z}{\gamma})+p(z-\theta)}
  - \frac{1}{2\gamma}|z|^2 +\inf\limits_{c\in\Cc} \brak{c-\frac{\alpha}{\gamma}c^\gamma e^{y}} - r - e_t + \frac{\delta_t}{\gamma}}.
\end{equation}
As in Lemma \ref{lemma1}, forall $\gamma \in (-\infty,0) \cup (0,1)$ and $Z \in \Pc, \argmin\limits_{p\in \Pc}\brak{\frac{1-\gamma}{2}|p_t|^2-g(p-\frac{Z}{\gamma})+p(Z-\theta)}$ is nonempty set and for any $p \in \argmin\limits_{p\in \Pc}\brak{\frac{\gamma-1}{2}|p_t|^2-g(p-\frac{Z}{\gamma})+p(Z-\theta)}$  we have a similar estimate as (\ref{estimatep}), namely 
\begin{equation}\label{estimateppower}
|p|^2\leq \beta_1|Z|^2+\beta_2\;\; dt\otimes d\p. a.e .
\end{equation} 
for some positive constant $\beta_1$ and $\beta_2$ regardless of $Z.$
In the same way, the set $\argmin\limits_{c\in\Cc}\brak{c-\frac{\alpha}{\gamma}c^\gamma e^{y}}$ is nonempty set and any $c\in \argmin\limits_{c\in\Cc}\brak{c-\frac{\alpha}{\gamma}c^\gamma e^{y}}$ is bounded as well as $c^\gamma.$
\\Note that $f(t,y,z)$ grows exponentially in $y$.
But it satisfies Assumption (A.1) in Briand and Hu \cite{BH08}.
So it can be deduced from Proposition 3 in \cite{BH08} that \eqref{BSDEpower} has a solution
$(Y,Z)$ such that $Y$ is bounded. That $Z$ is in ${\cal P}^{1 \times n}_{\rm BMO}$
and the uniqueness of such a solution then follow from \cite{Morlais}.

\begin{Theorem} \label{thmpower}
If $u(x) = x^{\gamma}/\gamma$ for $\gamma \in (-\infty,0) \cup (0,1)$, the
optimal value of the optimization problem \eqref{opt}
over all admissible strategies is
\be \label{optvaluepower}
\frac{1}{\gamma} x^{\gamma} e^{-Y_0},
\ee
and $(\tilde{c}^*, \tilde{p}^*)$ is an optimal admissible strategy if and only if
\be
\label{optstpower}
\tilde{c}^* \in \argmax_{\tilde{c} \in C}
\left(\frac{\alpha}{\gamma}\tilde{c}^\gamma e^{Y}-\tilde{c}\right)
\quad \mbox{and} \quad \tilde{p}^* \in \Pi_{P} \brak{\frac{Z + \theta}{1-\gamma}}.
\ee
In particular, an optimal admissible strategy exists, and it is unique
up to $\nu \otimes \p$-a.e. equality if the sets $C$ and $P$ are ${\cal P}$-convex.
\end{Theorem}
\begin{proof}

For every admissible strategy $(\tilde{c},\tilde{p})$ define the process
$$
R^{(p,c)}_t = \int_0^t \alpha e^{-\int_0^s \delta_udu} \dfrac{1}{\gamma}(c_sX^{(c,p)}_s)^{\gamma} ds + e^{-\int_0^t \delta_sds}\dfrac{1}{\gamma}  \brak{X^{(c,p)}_t}^{\gamma}e^{-Y_t}.
$$
Then
\begin{itemize}
\item For all $(p,c) \in \Acr \times \Ccr, R_T^{(p,c)} =\int_0^T \alpha e^{-\int_0^t \delta_sds} u(c_tX^{(c,p)}_t) dt + e^{-\int_0^T \delta_sds} u \brak{X^{(c,p)}_T}$,
\item $\forall t\in [0,T]$,  $\forall (p,c), (p',c') \in \Acr \times \Ccr $ such that $(p,c)=(p',c')$  on $[0,t]$, we have  $R_t^{(p,c)}=R_t^{(p',c')}$ $\P$.a.s
\end{itemize}
Using Ito formula, we obtain
$$
dR^{(p,c)}_t =  (-g(Z_t^{(p,c)})+A_t^{p,c})dt +Z_t^{(p,c)}dW_t,
$$
where
\beas
Z_t^{(p,c)}=e^{-\int_0^t \delta_udu} \brak{X_t^{(c,p)}}^\gamma e^{-Y_t}(p_t-\frac{1}{\gamma}Z_t),
\eeas
and
\beas
A^{(p,c)}_t = &&e^{-\int_0^t \delta_udu}\brak{X_t^{(c,p)}}^\gamma e^{-Y_t} \Big[ g(p_t-\frac{1}{\gamma}Z_t)+p_t (\theta_t-Z_t)+\frac{1}{2}(\gamma-1) |p_t|^2 +\frac{1}{2\gamma}|Z_t|^2+\frac{1}{\gamma} f(t,Y_t,Z_t)\\
&& + \frac{\alpha}{\gamma}c_t^\gamma e^{Y_t}
+ e_t - \tilde{c}_t + r - \frac{\delta_t}{\gamma}\Big].
\eeas
\beas
\E[\dint_0^T |Z_t^{(p,c)}|^2dt]&=\E[\dint_0^T \brak{X_t^{(c,p)}}^{2\gamma} e^{-2Y_t}|p_t+\frac{1}{\gamma}Z_t|^2dt] \\
& \leq K \E[\sup\limits_{0 \leq t \leq T} \brak{X_t^{(p,c)}}^{2\gamma} \dint_0^T |p_t+\frac{1}{\gamma}Z_t|^2dt]\\
& \leq K \E[\sup\limits_{0 \leq t \leq T} \brak{X_t^{(p,c)}}^{q\gamma}]^{\frac{2}{q}} \E[\dint_0^T |p_t+\frac{1}{\gamma}Z_t|^2dt)^{\frac{q}{q-2}}]^{\frac{q-2}{q}}<\infty.
\eeas
\beas
\E[(R^{(p,c)}_T)^2] & \leq C \E[\sup\limits_{0 \leq t \leq T} \brak{X_t^{(c,p)}}^{\gamma}(1+\int_0^Tc_t^\gamma)] \\
& C \E[\sup\limits_{0 \leq t \leq T} \brak{X_t^{(c,p)}}^{q\gamma}]^\frac{1}{q}\brak{1+[\E\brak{\dint_0^Tc_t^\gamma}^{\frac{q}{q-1}}]^{\frac{q-1}{q}}}\leq \infty.
\eeas
So, for all admissible strategy $(p,c),$ we have $\int_0^T A^{(p,c)}_t dt \in L^2(\Omega, \Fc,\p).$
By definition of the function $f$, the process $A_t^{(p,c)}$ is negative a.e, and so forall admissible strategy $R^{(p,c)}$ is $g$-supermartingale.
\\On the other hand, for $p^* \in \argmin\limits_{p\in \Pc}\brak{\frac{\gamma-1}{2}|p_t|^2-g(p-\frac{Z}{\gamma})+p(Z-\theta)}$ and $c^* \in \argmin\limits_{c\in\Cc}\brak{c-\frac{\alpha}{\gamma}c^\gamma e^{y}}$, we have $A^{(p^*,c^*)}_t=0\quad dt \otimes d\p.a.e$ and so $R^{(p^*,c^*)}$ is $g$-martingale.
It remains to be verified that $(p^*,c^*)$ is an admissible strategy. Note that $(c^*)^{\gamma}$ is a continuous bounded process, then $\E[(\dint_0^T (c^*_t)^\gamma dt)^{\frac{q}{q-1}}] <\infty.$ Using estimation(\ref{estimateppower}), we deduce that  $\E[(\dint_0^T |p^*_t|^2 dt)^{n}] <\infty$ forall $n\geq 1.$
To complete the proof,we have  for all $q > 2$, 
\beas
 (X^{(p^*,c^*)}_t)^{\gamma}&& = x^{\gamma}e^{\displaystyle{\gamma \int_0^t p_s^*(dW_s+\theta_sds)-\frac{\gamma}{2} \int_0^t |p^*_s|^2ds + \gamma \int_0^t(r+e_s-c^*_s)ds }}\\
&& \leq d_1 e^{\displaystyle{\gamma \int_0^t p_s^*dW_s+\frac{\gamma(\gamma-1)}{2}\int_0^t|p_s^*|^2ds}}\\
&& \leq d_1 e^{\displaystyle{\gamma \int_0^t p_s^*dW_s-\gamma^2 \int_0^t|p_s^*|^2ds}} \times e^{\displaystyle{\frac{\gamma(3\gamma-1)}{2}\int_0^t|p_s^*|^2ds}}\\
&& \leq d_2 \brak{e^{\displaystyle{2\gamma \int_0^t p_s^*dW_s-2\gamma^2 \int_0^t|p_s^*|^2ds}} + e^{\displaystyle{\gamma(3\gamma-1)\int_0^t|p_s^*|^2ds}}}
\eeas
and so 
\begin{equation}\label{maxestimatep}
\sup\limits_{0\leq t \leq T} (X^{(p^*,c^*)}_t)^{\gamma} \leq d_2 \brak{\sup\limits_{0\leq t \leq T}\Ec(  \int_0^t 2\gamma p^*_sdW_s) +\sup\limits_{0\leq t \leq T} e^{\displaystyle{\gamma(3\gamma-1)\int_0^t|p_s^*|^2ds}}}\;\;\; \p.a.s
\end{equation}
The continuous process $\brak{  \dint_0^t 2\gamma p^*_sdW_s}_{0\leq t \leq T}$  is a BMO martingale then the process $M_.=\brak{\Ec(\dint_0^. 2\gamma p^*_sdW_s)}_{0\leq t \leq T}$ is uniformly integrable martingale.  By Doob’s martingale maximal inequalities, we have, forall $q>1$
$$\E[(\sup\limits_{0\leq t \leq T}|M_t|)^q] \leq (\dfrac{q}{q-1})^q\E(|M_T|^q).$$
\beas
 (M_T)^{q}&& = e^{\displaystyle{2q\gamma \int_0^T p_s^*dW_s-2q\gamma^2 \int_0^T|p_s^*|^2ds}} \\
&&= e^{\displaystyle{2q\gamma \int_0^T p_s^*dW_s-4q^2\gamma^2 \int_0^T|p_s^*|^2ds}}\times e^{\displaystyle{2q(2q-1)\gamma^2 \int_0^T|p_s^*|^2ds}} \\
&& \leq \dfrac{1}{2}\brak{e^{\displaystyle{4q\gamma \int_0^T p_s^*dW_s-8q^2\gamma^2 \int_0^T|p_s^*|^2ds}}+ e^{\displaystyle{4q(2q-1)\gamma^2 \int_0^T|p_s^*|^2ds}}}
\eeas
Then, $\E(|M_T|^q) < \infty.$
Using Lemma\ref{lemma1} and estimation \ref{YZestimation}, there a positive constant $M_1$ such that  for all  
$$\forall k \in \mathbb{N} \setminus \{0\},\;\; \mathbb{E}\left[\left(q\int_0^T\left|p^*_s\right|^2 d s\right)^{k}\right]<M_1q^k.$$
Which implies that
$$\forall n \in \mathbb{N} \setminus \{0\},\;\; \mathbb{E}\left[\sum\limits_{k=1}^n\frac{1}{k!}\left(q\int_0^T\left|p^*_s\right|^2 d s\right)^{k}\right]<{M_1}(e^q-1).$$
By Fatou lemma we deduce that 
$$\E[e^{q\int_0^T |p^*|^2_s  ds}]< M_1(e^q-1)+1< \infty.$$
Using convex inequality in \ref{maxestimatep}, we obtain 
\begin{equation}
\sup\limits_{0\leq t \leq T} e^{-\gamma p  h(t) X_t} \leq d_3 2^{p-1} (\sup\limits_{0\leq t \leq T}\Ec(-2 p \gamma  \int_0^t h(s)p^*_sdW_s) + e^{ p \int_0^T |p^*|^2_s  ds})\;\;\; \p.a.s.
\end{equation}
Which means that $\E[\sup\limits_{0\leq t \leq T} e^{-\gamma p  h(t) X_t}] < \infty.$
\end{proof}
\begin{Remark}
If there is no constraint on $c$ ($\Cc= \Pc^{1\times 1}$), then the infimum $\inf\limits_{c\in\Cc} \brak{c-\frac{\alpha}{\gamma}c^\gamma e^{y}}$ is equal $(1-\dfrac{1}{\gamma})\alpha^{\frac{1}{1-\gamma}}e^{\frac{1}{1-\gamma}y}$ and it is rached on $\underline{c}=\alpha^{\frac{1}{1-\gamma}}e^{\frac{1}{1-\gamma}y}$. So, the optimal consumption strategy is given by $$c^*=\alpha^{\frac{1}{1-\gamma}}e^{\frac{1}{1-\gamma}Y},$$
where $(Y,Z)$ is solution of BSDE (\ref{BSDEpower}).
\end{Remark}
\begin{Example}
In the case of linear generator $g(t,z)=\eta_t.z$, the driver $f$ is given by
$$
f(t,y,z)
= \gamma \brak{\frac{1-\gamma}{2}{\rm dist_t^2}(p, \frac{\theta+\eta-z}{1-\gamma})+\frac{\eta_t}{\gamma}.z- \frac{1}{2(\gamma-1)}|\theta+\eta-z|^2
  - \frac{1}{2\gamma}|z|^2 +\inf\limits_{c\in\Cc} \brak{c-\frac{\alpha}{\gamma}c^\gamma e^{y}} - r - e_t + \frac{\delta_t}{\gamma}}.
$$
So, by setting $\eta\equiv 0$, we find the results of Cheridito et all \cite{cheridito11}.
If there is no constraint, we have ${\rm dist_t^2}(p, \frac{\theta+\eta-z}{1-\gamma})=0$ and this is achieved on $p= \frac{\theta+\eta-z}{1-\gamma}.$ So $p^*=\frac{\theta+\eta-Z}{1-\gamma}$ where $(Y,Z)$ is solution of BSDE (\ref{BSDEpower}).

\end{Example}
\begin{Example}{The $\kappa$- ignorance utility}\\
We consider the one-dimensional context, and we assume there is no constraint on trading strategy ($\Qc=\Pc$). For all $(t,\omega)$, we consider the functionnal
$$\ell: x \longmapsto \frac{1-\gamma}{2}|x|^2 - \kappa |x-\dfrac{z}{\gamma}| + p(z-\theta_t(\omega)).$$
If $z \neq \gamma \theta(\omega)$, $\ell$ admits a global minimum given by 
 $\dfrac{z-\theta_t (\omega)-\kappa}{\gamma-1}\textbf{1}_{ z < \gamma \theta_t(\omega)}+\dfrac{z-\theta(\omega)+\kappa}{\gamma-1}\textbf{1}_{ z > \gamma \theta_t(\omega)}$ if $0 < \gamma < 1; $ and 
$ \dfrac{z-\theta_t(\omega)-\kappa}{\gamma-1}\textbf{1}_{ z > \gamma \theta_t(\omega)}+\dfrac{z-\theta_t(\omega)+\kappa}{\gamma-1}\textbf{1}_{ z < \gamma \theta_t(\omega)} $ if $ \gamma < 0$. In the case where $z = \gamma \theta_t(\omega),$ $\ell$ admits two global minimums given by $\dfrac{z-\theta_t(\omega)\pm\kappa}{\gamma-1}.$ So an optimal strategy can be written as follows,
\begin{equation*}
\begin{array}{lll}
p^* =\dfrac{1}{\gamma-1} (Z-\theta- \kappa\textbf{1}_{ Z < \gamma \theta}+\kappa\textbf{1}_{ Z > \gamma \theta}+ \zeta \textbf{1}_{ Z = \gamma \theta} & {\rm if} & 0 < \gamma < 1; \\
p^* =\dfrac{1}{\gamma-1} (Z-\theta- \kappa\textbf{1}_{ Z > \gamma \theta}+\kappa\textbf{1}_{ Z < \gamma \theta}+ \zeta \textbf{1}_{ Z = \gamma \theta} & {\rm if} & \gamma < 0.
\end{array}
\end{equation*}
where $\zeta=(\zeta_t)_{t\in [0,T]}$ is a real predictable process taking value in $\{-\kappa, \kappa\}.$ Note that we do not necessarily have the uniqueness of the optimal strategy.

\end{Example}

\section{Logarithmic utility}\label{sec:log}
\label{subsec:log}

In the case $u(x) = \log(x)$, the wealth process is the same as in the previous section namely
\be
X^{(c,p)}_t =
x \, {\cal E} \brak{\tilde{p} \cdot W^{\q}}_t \exp\brak{\int_0^t( r
+ \tilde{e}_s - \tilde{c}_s) ds} > 0,
\ee
where the trading strategies $(\tilde{\pi}, \tilde{c})$ are constrained to take values in a sequentially closed and ${\cal P}$-stable set  $\Qc \times \Cc  \subseteq {\cal P}^{1 \times m}\times {\cal P}$. 
We make the following assumption:
\be \label{asslog}
\mbox{there exists a pair }
(\bar{p},\bar{c}) \in \Qc \times \Cc \mbox{ such that } \bar{p} \mbox{ and }
u_{\gamma}(\bar{c}) - \bar{c}   \mbox{ are bounded.}
\ee
\begin{Definition}
For $u(x) = \log(x)$, an admissible strategy is a pair $(c,p) \in C \times P$
satisfying
\be \label{condlog}
\mathbb E\left[\int_0^T|\log( c_t)|^2 dt+\int_0^T c_t dt+\int_0^T
|p_t|^2 dt\right] < \infty.
\ee
\end{Definition}
Remember that we understand $\log(x)$ to be $-\infty$ for $x \le 0$. Therefore,
\eqref{condlog} implies $c > 0$ $\nu \otimes \p$-a.e.
Note that this implies that $u_{\gamma}(\bar{c})$ and $\bar{c}$ are both bounded. 
\begin{Lemma}\label{lemma2}
\begin{enumerate}
\item
$\forall Z\in \Qc; \;\; \argmin\limits_{p\in \Qc}\brak{\frac{1}{2} |p_t|^2 -p_t\theta_t- g(p_t-Z)}$ is non empty subset from $\Qc$.
 Moreover, there are two constants $\kappa_3$ and $\kappa_4$ such that  $\forall  Z\in \Qc, \forall  p\in \argmin\limits_{p\in \Qc}\brak{\frac{1}{2} |p_t|^2 -p_t\theta_t- g(p_t-Z)}$ we have  
\begin{equation}\label{plogestimate}
|p|\leq \kappa_3|Z|+\kappa_4\;\; dt\otimes d\p. a.e .
\end{equation}
\item $\argmin\limits_{c\in C} \left(c-\frac{\alpha}{h}\ln(c)\right)$ is non empty subset from $\Cc$
and $\min\limits_{c\in C} \left(c-\frac{\alpha}{h}\ln(c)\right)$ is bounded.
\end{enumerate}
\end{Lemma}

\begin{proof}
\begin{enumerate}
\item

Let $f_Z:\Ac \longrightarrow \bar{L}; p\mapsto \frac{1}{2} |p|^2 -p\theta- g(p-Z).$
\\$f_Z$ is sequentially continuous stable and the set 
$$\Gamma:=\{p\in \Qc\;\;s.t\;\; f(p)\leq f(\bar{p})\}$$
is $L^0-$bounded, indeed, let $p \in \Gamma$ we have
$$
\begin{aligned}
\frac{1}{2} |p|^2 -p\theta- g(p-Z) \leq \frac{1}{2} |\bar{p}|^2 -\bar{p}\theta- g(\bar{p}-Z)
\end{aligned}
$$
Using Cauchy Schwartz inequality and assumption on $g$, we obtain
$$
\begin{aligned}
\frac{1}{2}|p|^2- |p||\theta| - \mu(|p-Z|+1) \leq \frac{1}{2}|\bar{p}|^2+|\bar{p}||\theta| + \mu(|\bar{p}-Z|+1)
\end{aligned}
$$
this implies
$$
\begin{aligned}
\frac{1}{2}|p|^2- (\mu+\kappa)|p| \leq   \frac{1}{2}|\bar{p}|^2+(\mu+\kappa)|\bar{p}| +2\mu|Z|+ 2\mu
\end{aligned}
$$
and so,

\begin{equation}\label{estimatep3}
|p|\leq  \frac{1}{2}|p|^2- (\mu+\kappa)|p| + \frac{1}{2}(\mu+\kappa+1)^2  \leq   2\mu|Z|+\frac{1}{2}|\bar{p}|^2+(\mu+\kappa)|\bar{p}| + 2\mu+\frac{1}{2}(\mu+\kappa+1)^2.
\end{equation}

Hence $\Gamma$ is $L^0-$bounded. Consequently, using Theorem 4.4 in \cite{cheridito11}, there exists at least $\hat{p} \in \Qc$ such that $\essinf\limits_{p\in \Qc}f_Z(p)=f_Z(\hat{p}).$.
\\
Estimate \ref{plogestimate} is an immediate consequence of the inequality \ref{estimatep3} and  the fact that  $\bar{p}$ are bounded.
\item The map $c \mapsto c-\frac{\alpha}{h}\ln(c)$ is is sequentially continuous stable and the set 
$$\{c \in \Cc \;\; s.t \;\; c-\frac{\alpha}{h}\ln(c) \leq \bar{c}-\frac{\alpha}{h}\ln(\bar{c})\}$$
is is $L^0-$bounded. Moreover
$$\frac{\alpha}{h}(\ln(\frac{\alpha}{h})-1)\leq \min\limits_{c\in C} \left(c-\frac{\alpha}{h}\ln(c)\right)\leq \left(\bar{c}-\frac{\alpha}{h}\ln(\bar{c})\right)$$
\end{enumerate}
\end{proof}
Consider now  the BSDE
\be \label{BSDElog}
Y_t =  \int_t^T f(s,Y_s,Z_s)ds - \int_t^T Z_s dW_s
\ee
with driver
\begin{equation}\label{driverlog}
f(t,y,z) = \dfrac{-\alpha y}{h(t)}+ \essinf\limits_{p\in P}\brak{\frac{1}{2} |p|^2 -p\theta- g(p-z)}_t + \essinf\limits_{c\in C}\left(c-\frac{\alpha}{h} \log(c)\right)_t
 -e_t - r ,
\end{equation} 
where $$
h(t) = \alpha \int_t^Te^{-\int_t^s \delta_udu}ds+ e^{-\int_t^T \delta_udu}
$$ 
is the  unique solution on the Cauchy problem 
$$h'(t)=\delta_t h(t)-\alpha \quad \mbox{with} \quad h(T)=1.$$
Note  that $f(t,y,z)$ is of linear growth in $y$, and quadratic in $z$. Since $\theta$, $e$, $E$ and $h$ are bounded, using Lemma \ref{lemma2}, there exists a positive constant $K $ such that
$$|f(t,y,z)| \le K(1+ |y| + |z|^2)$$
and
$$|f(t,y_1,z_1)-f(t,y_2,z_2)|\le K( |y_1-y_2|+(1+|z_1|+|z_2|)|z_1-z_2|).$$
So it follows from Kobylanski \cite{Kobylanski} that equation \eqref{BSDEexp} has a unique solution
$(Y,Z)$ such that $Y$ is bounded and from  Morlais \cite{Morlais} that $Z$
belongs to ${\cal P}^{1 \times n}_{\rm BMO}$.
According Briand and Hu \cite{BH08}, $(Y,Z)$ satisfies for each $p > 1$
\begin{equation}\label{YZestimationl}
\mathbb{E}\left[\exp \left(\gamma p \sup _{0 \leq t \leq T}\left|Y_t\right|\right)+\left(\int_0^T\left|Z_s\right|^2 d s\right)^{p / 2}\right] \leq C \mathbb{E}\left[\exp \left(2p K\left(|F|+KT\right)\right)\right]<M. 
\end{equation}
where $M$ depends on $p, K , \|h\|_{\infty}, \|F\|_{\infty}$ and $T$.
\begin{Theorem} \label{thmlog}
For $u(x) = \log(x)$, the optimal value of the optimization problem
\eqref{opt} over all admissible strategies is
\be \label{optvaluelog}
h(0)(\log(x)-Y_0),
\ee
and $(\tilde{c}^*, \tilde{p}^*)$ is an optimal admissible strategy if and only if
\be
\label{optstlog}
\tilde{c}^* \in \argmax_{\tilde{c}\in C} \left(\frac{\alpha}{h} \log(\tilde{c})-\tilde{c}\right)
\quad \mbox{and} \quad \tilde{p}^* \in \Pi_{P} \brak{\theta}.
\ee
In particular, an optimal admissible  strategy exists, and it is unique
up to $\nu \otimes \p$-a.e. equality if the sets $C$ and $P$ are ${\cal P}$-convex.
\end{Theorem}

\begin{proof}
For every admissible strategy $(\tilde{c},\tilde{p})$, define the process
$$
R^{p,c}_t = h(t) e^{- \int_0^t \delta_u du } \brak{\log \brak{X_t^{(p,c)}}-Y_t}+\int_{0}^t \alpha e^{- \int_0^s \delta_u du } \log(c_sX_s^{p,c})ds
$$
Then
\begin{itemize}
\item For all $(p,c) \in \Acr \times \Ccr, R_T^{(p,c)} = e^{- \int_0^T \delta_u du } \brak{\log \brak{X_T^{(p,c)}}}+\int_{0}^T \alpha e^{- \int_0^s \delta_u du } \log(c_sX_s^{c,p})ds$,
\item $\forall t\in [0,T]$,  $\forall (p,c), (p',c') \in \Acr \times \Ccr $ such that $(p,c)=(p',c')$  on $[0,t]$, we have  $R_t^{(p,c)}=R_t^{(p',c')}$ $\P$.a.s
\end{itemize}
and
\be \label{diffRlog}
dR^{p,c}_t = -g(Z_t^{p,c}) dt + A^{p,c}_t dt + Z_t^{p,c} dW_t ,
\ee
where
\beas
Z_t^{p,c}=h(t)e^{- \int_0^t\delta_udu} (p_t-Z_t), 
\eeas
and
\beas
A^{p,c}_t = && h(t)e^{- \int_0^t\delta_udu}  \edg{ g(p_t-Z_t)+p_t\theta_t-\frac{1}{2}|p_t|^2
+\frac{\alpha Y_t}{h(t)}  + f(t,Y_t,Z_t)+ \frac{\alpha \log(c_t)}{h(t)} -c_t+e_t+
r }.
\eeas
It follows from the definition of admissible strategy that $R_T^{p,c}, \int_0^Tg(Z_t^{p,c})dt$ and $\int_0^T Z_t^{p,c}dW_t$  are both in $L^2(\Omega,\Fc,\p)$ and so $\E[(\int_0^TA^{p,c}_t dt)^2]< \infty.$ Moreover, by definition of the function $f$, we have $A_t^{p,c}\leq 0 dt\times d\p.a.e.$ So for all admissible strategy $(p,c)$, $R^{(p,c)}$ is $g-$supermartingale.
\\ Let $p^* \in \argmax\limits_{p\in P}\brak{\frac{1}{2} |p_t|^2 -p_t\theta_t- g(p_t-Z)}$ and $c^* \in \argmin\limits_{c\in C}\left(c-\frac{\alpha}{h} \log(c)\right)$ . Using Lemma \ref{lemma2}, we obtain
$\E[\int_0^T|p^*_t|^2dt]\leq \kappa_1\E[\int_0^T|Z_t|^2dt]+ \kappa_2 < \infty$ and 
$\mathbb E\left[\int_0^T|\log( c^*_t)|^2 dt+\int_0^T c^*_t dt\right] < \infty$ because $\dfrac{\alpha}{h}\ln(c^*)-c^*$ is bounded. Which guarantees that $(p^*,c^*)$ is admissible strategy.
\end{proof}

\begin{Example}
In the case of linear generator $g(t,z)=\eta_t.z$, the driver $f$ is given by
$$
f(t,y,z) = \dfrac{-\alpha y}{h(t)}+ \dfrac{1}{2}{\rm dist}_t^2(P,\theta+\eta)-\dfrac{1}{2}|\theta+\eta|^2_t+\eta_t.z + \essinf\limits_{c\in C}\left(c-\frac{\alpha}{h} \log(c)\right)_t -e_t - r ,
$$
So, by setting $\eta\equiv 0$, we find the results of Cheridito et all \cite{cheridito11}.
If $\dfrac{\alpha}{h} \in \Cc$ and $\theta+\eta \in P$, (this is the case, for example, where no constraints on investment consumption are required), then 
$${\rm dist}_t(P,\theta+\eta)\equiv 0,\quad
\essinf_{\tilde{c}\in C} \left(\tilde{c}-\frac{\alpha}{h}
\log(\tilde{c}) \right)= \frac{\alpha}{h}
\brak{\log \brak{1-\frac{\alpha}{h}}}.$$ So, the generator $f$ can be written 
 $$
f(t,y,z) = \dfrac{-\alpha y}{h(t)}+ -\dfrac{1}{2}|\theta+\eta|^2_t+\eta_t.z + \frac{\alpha}{h}\brak{\log \brak{1-\frac{\alpha}{h}}}_t -e_t - r ,
$$
and optimal investment consumption strategy is given by
$(p^*,c^*)= (\theta+\eta,\dfrac{\alpha}{h}).$ Moreover, the optimal value is given by
$$v(0)=h(0)(\ln(x)-\E[ \int_0^T \Gamma_s^0 \varphi_s d s ]).$$

where for $s \geq t\;\;\; \Gamma_s^t$ is given by

$$
\Gamma_s^t=\exp(\int_t^s(\alpha h^{-1}(s)-\frac{1}{2}\eta_s^2)ds +\int_t^s\eta_sdW_s)
$$
and
$$
\varphi_t= \dfrac{1}{2}|\theta+\eta|^2_t+  \frac{\alpha}{h}\brak{\log \brak{1-\frac{\alpha}{h}}} +e_t + r.
$$
The optimal investment strategy $p^*$ is equal to $Z+\frac{\theta-\eta}{\gamma}$ and 
\end{Example}
\begin{Example}{The $\kappa$- ignorance utility}\\
We consider the one-dimensional context, and we assume there is no constraint on trading strategy ($\Qc=\Pc$). For all $(t,\omega)$, we consider the functionnal
$$\ell: x\longmapsto \frac{1}{2} |x|^2 -x\theta_t(\omega)- \kappa|x-z|.$$
If $\theta_t(\omega) \neq  z$, $\ell$ admits a global minimum  given by $\theta_t(\omega)+\kappa$ if $z < \theta (\omega) $ and $\theta_t(\omega)-\kappa$ if $z > \theta (\omega)$. In the case where $\theta_t(\omega) =  z$, $\\ell$ admits two global minimums given by $\theta_t(\omega)\pm\kappa.$ So, an optimal strategy can be written
We obtain at least two optimal investment strategies given by
$$p^*=\theta+\kappa\textbf{1}_{Z < \theta} -\kappa\textbf{1}_{Z > \theta}+\zeta \textbf{1}_{Z = \theta} ,$$
where $\zeta=(\zeta_t)_{t\in [0,T]}$ is a predictable process taking value in $\{-\kappa, \kappa\}.$ Note that we do not necessarily have the uniqueness of the optimal strategy.
\end{Example}
\begin{Remark}\label{remcarra}
Thanks estimate\ref{YZestimationl}, we can assume that $\tilde{e}$ is exponentially integrable instead of bounded.
\end{Remark}
\section{Conclusion}
This study enhances the understanding of utility maximization in the context of incomplete financial markets by integrating nonlinear expectations and general constraints. 
In the case of exponential utility, we assumed that the income rate $e$ and the final payment $E$ are bounded and without any constraints on consumption $c.$ The proof of Theorem \ref{thmexp} relies on fondamental estimation (\ref{YZestimation})  of a unique solution $(Y,Z)$ to the BSDE equation (\ref{BSDEexp}). This estimate, which will make it possible to overcome the difficulties linked to the use of the bmo characters of $Z.$, will also make it possible to suppose that e and F are not bounded as is shown in the remark \ref{remexp}.
. In the carra utility situations, we have assumed that $E = 0$ and $\tilde{e} = e/X$ is bounded, but the results remain valid when $\tilde{e}$ is exponentially integrable as shown in the remark \ref{remcarra}. 
Through examples, we explicitly give the expressions of the optimal investment-consumption strategies in various cases of the utility functions studied. 
 
\end{document}